\documentclass{amsart}

\usepackage{enumitem}
\usepackage[charter]{mathdesign}
\usepackage{amsmath}                

\usepackage{amsfonts}
\usepackage{graphicx, subfigure}
\usepackage{graphics}
\usepackage{epstopdf}
\usepackage{amsthm}

\usepackage[usenames, dvipsnames, pdftex]{color}
\usepackage[colorlinks=true,
        raiselinks=true,
        linkcolor=MidnightBlue,
        citecolor=ForestGreen,
        urlcolor=Mahogany,
        plainpages=false]{hyperref}

\usepackage{tikz}
\usepackage[all]{xy}
\usetikzlibrary{arrows,automata}
\usetikzlibrary{shapes,snakes}

\usepackage{mathrsfs}

  {
      \theoremstyle{plain}
      \newtheorem{theorem}{Theorem}
      \newtheorem{lemma}{Lemma}

      \newtheorem{assumption}{Assumption}
      \newtheorem{example}{Example}
      \newtheorem{remark}{Remark}

  }

\def\N{\mathbb{N}}
\def\RR{\mathbb{R}_{\geq 0}}

\newcommand{\Ex}{\mathsf{E}}
\newcommand{\pr}{\mathsf{P}}

\def\d{\mathrm{d}}
\def\F{\mathbb{F}}
\def\R{\mathbb{R}}
\def\N{\mathbb{N}}

\def\a{\mathfrak{A}}
\def\f{\mathfrak{F}}
\def\i{\mathfrak{I}}
\def\l{\mathfrak{L}}
\def\b{\mathfrak{B}}
\def\B{\mathrm{b}\b}

\def\tz{\tilde{z}}

\def\ve{\varepsilon}

\DeclareMathOperator{\arctanh}{arctanh}

\title{Computation of ruin probabilities\\ for general discrete-time Markov models}

\author{Ilya Tkachev and Alessandro Abate
        }
\thanks{
    \hspace{-0.6cm}
    I. Tkachev is with the Delft Center for Systems \& Control, Delft University of Technology, The Netherlands. Email: \texttt{i.tkachev@tudelft.nl}. \\
    A. Abate is with the Department of Computer Science, University of Oxford, United Kingdom, and with the Delft Center for Systems \& Control, Delft University of Technology, The Netherlands. Email: \texttt{alessandro.abate@cs.ox.ac.uk}.
}

\begin{document}
\maketitle

\begin{abstract}
  We study the ruin problem over a risk process described by a discrete-time Markov model.
  In contrast to previous studies that focused on the asymptotic behaviour of ruin probabilities for large values of the initial capital,
  we provide a new technique to compute the quantity of interest for any initial value, and with any given precision.
  Rather than focusing on a particular model for risk processes,
  we give a general characterization of the ruin probability by providing corresponding recursions and fixpoint equations.
  Since such equations for the ruin probability are ill-posed in the sense that they do not allow for unique solutions,
  we approximate the ruin probability by a two-barrier ruin probability,
  for which fixpoint equations are well-posed.
  We also show how good the introduced approximation is by providing an explicit bound on the error and by characterizing the cases when the error converges to zero.
  The presented technique and results are supported by two computational examples over models known in the literature, one of which is extremely heavy-tailed.

\end{abstract}



\pagestyle{myheadings}
\thispagestyle{plain}
\markboth{I. Tkachev and A. Abate}{Computation of ruin probabilities for Markov models}

\section{Introduction}
The ruin problem,
one of the most fundamental problems in risk theory,
studies the long-term behavior of a stochastic process that represents the evolution of the capital of an insurance company (shortly, a risk process).
The objective is to estimate the probability that the value of the risk process at some time becomes negative \cite{a2000}.
The explicit solution for this problem is available only in a limited number of instances,
even for simple models of the risk process such as the Cramer-Lundberg model that is given by a random walk \cite{m2009}.
Owing to this limitation,
the usual approach to the problem in the literature is to focus on a particular model
and to characterize its ruin probability by deriving upper bounds on its value \cite{c2002,cd2004,dr2009}.
However, the derivation of such bounds critically depends on whether the distribution of the claim size is heavy-tailed,
and the quality of most of the bounds is practically useful only when the initial capital is large.

This paper focuses on the ruin problem over risk processes described by general discrete-time Markov models,
which represent a rather rich class of models studied e.g. in \cite{c2002}, \cite{cd2004} and \cite{dr2009}.
This class can also encompass the ruin problem for a number of continuous-time models which,
since the structure of a risk process is essentially characterized by jumps,
can be equivalently solved for their discrete-time version \cite{m2009,p2008}.

The contribution of this paper is twofold and concerns both the characterization and the computational aspects related to the ruin problem.
Firstly, we consider a general Markov model that includes instances known in the risk theory literature as special cases.
In this framework we apply (theoretical and computational) methods from formal verification \cite{bk2008},
which have been recently developed in \cite{ta2012a-arXiv},
to derive the fixpoint equations for the ruin probability and for the related two-barrier ruin probability, and further elaborate on their properties.
The ruin problem and the two-barrier ruin problem are known respectively as the reachability and reach-avoid problem (or constrained-reachability problem) in formal verification \cite{aklp2010,rcsl2010}.
In particular, we show that the fixpoint equation for the ruin probability has a unique solution if and only if the solution is constant.
Secondly, we show how to use known upper bounds on the ruin probability to compute its value using the two-barrier problem with any given precision and for any value of the initial capital.
An important feature of this approach is that the tightness of the aforementioned bounds does not affect the quality of the approximation.
This feature allows the application of our technique to ill-behaved cases that are characterized by heavy-talied distributions,
whereas the generality of the approach allows dealing with complicated models in a unified way.

\smallskip

The work is structured as follows:
Section \ref{ssec:notation} describes notations and conventions used in this paper and Section \ref{ssec:models} introduces the models.
Section \ref{ssec:ruin-inv} continues with the introduction of the ruin probability problem for a general Markov model of the risk process,
its analysis and the presentation of methods for its solution.
Section \ref{ssec:CL} applies the results to the Cramer-Lundberg model (Example \ref{ex:1}),
whereas Section \ref{ssec:CL-I} to a model with interest rates (Example \ref{ex:2}).
Based on these models,
Section \ref{sec:cs} finally provides two case studies
(one of which is extremely heavy-tailed)
to display the computational viability of the presented methods and the improvements on results from the literature.

\section{Theoretical results}\label{sec:theor}

\subsection{Notations and conventions}\label{ssec:notation}

In order to provide a general characterization of ruin probabilities,
we introduce some notions from the theory of discrete-time Markov processes \cite{r1984}.
We focus on processes with Borel state spaces:
a topological space $E$ is called a Borel space if it is homeomorphic to a Borel subset of a complete separable metric space \cite{hll1996}.
The Borel $\sigma$-algebra of $E$ is denoted by $\b(E)$.
Examples of Borel spaces are the Euclidian space $\R^n$, any of its Borel subsets,
in particular the set of natural numbers $\N$,
or hybrid state spaces \cite{aklp2010}.
We further put forward the following notation $\N_0 := \N\cup\{0\}$ and $\RR := [0,\infty)$.

Let $(\Omega,\f,\pr)$ be some probability space.
A discrete-time Markov process on $E$ is a sequence of random variables $X = (X_n)_{n \in \N_0}$, where $X_n:(\Omega,\f)\to(E,\b(E))$
which satisfies the Markov property,
namely
\begin{equation*}
  \pr(X_{n+1}\in A|\f_n) = \pr(X_{n+1}\in A|X_n)
\end{equation*}
for all $A\in \b(E)$ and $n \in \N_0$.
Here by $\F := \{\f_n\}_{n\in \N_0}$ we denote the natural filtration of $X$,
that is $\f_n := \sigma\{X_k,0\leq k\leq n\}$.
For any $x\in E$ we denote $\pr_{x}(\cdot) = \pr(\cdot|X_0 = x)$ and by $\Ex_x$ the corresponding expectation. Finally, the transition kernel of $X$ is a stochastic kernel on $E$ given by $P(x,A) := \pr_x(X_1\in A)$ for all $x\in E$ and $A\in \b(E)$.\footnote{For a comprehensive introduction to general discrete-time Markov processes, and in particular for the rigorous construction of the measure $\pr_x$, the reader is referred to \cite[Chapter 2]{r1984}.}

\smallskip

The space of all bounded Borel measurable functions $f:E\to\R$ is denoted by $\B(E)$; it is a Banach space with a norm given by $\|f\| := \sup_{x\in E}|f(x)|$. For a linear operator $\a$ on $\B(E)$ its induced norm is given by $\|\a\| = \sup_{\|f\|\leq 1} \|\a f\|$. Whenever it holds that $\|\a\| <1$ we say that the operator $\a$ is a contraction. As an example, a transition kernel $P$ induces a linear operator on $\B(E)$ defined as follows:
\begin{equation*}
  Pf(x) := \Ex_x[f(X_1)] = \int_E f(y)P(x,\d y).
\end{equation*}
Clearly, $\|P\| = 1$ and hence it is not a contraction.
An additional relevant operator for this work is the invariance operator,
which is defined for every set $A\in \b(E)$ as
\begin{equation*}
  \i_Af(x) =1_A(x)Pf(x),
\end{equation*}
where $1_A$ denotes the indicator function of set $A$.
The possible contractivity of the operator $\i_A$ plays an important role below and depends on the set $A$.

\subsection{Models of risk processes}\label{ssec:models}
In this work we consider discrete-time risk models described by a process $X = (Z,\theta)$,
which is made up of a real-valued surplus process $Z$ and a process $\theta$, taking values in some Borel space $\Theta$,
which describes the evolution of parameters.
At any time instance $n \in \N_0$ the capital (of, say, a company) is $Z_n$,
whereas the value of the parameters of the model is given by $\theta_n$.
The parameters process $\theta$ can represent quantities such as the interest rate,
the return rate on risky investments,
the historical values of $Z$, etc. (see examples below).
Without loss of generality \cite[Proposition 8.6]{k1997a} we can assume that the Markov process $X = (Z,\theta)$ is described recursively in the following form:
\begin{equation}\label{model:general}
  \begin{cases}
    Z_{n+1} &= g(Z_n,\theta_n,\xi_n),\quad Z_0 = z,
    \\
    \theta_{n+1} &= h(Z_n,\theta_n,\xi_n),\quad \theta_0 = \theta,
  \end{cases}
\end{equation}
where $\xi$ is a sequence of i.i.d. random variables $\xi_n$ with values in a Borel space $\Xi$,
and functions $g$ and $h$ are measurable.
The state space for \eqref{model:general} is $E = \R\times \Theta$,
the sample space $\Omega = \Xi^{\N_0}$ is the space of infinite sequences over $\Xi$,
and $\F$ is the natural filtration of the process $(\xi_n)_{n \in \N_0}$.

We use $E$ or $\R\times\Theta$ as well as $X$ or $(Z,\theta)$ interchangeably, with the following convention:
when it is important to focus on general Markov structure of the process we use the first notation;
we otherwise employ the latter when we intend to emphasize the role of process $Z$.
In order to clarify the notations,
we provide examples of two discrete-time models taken from the literature,
and discuss their expression in the form of \eqref{model:general}.
The first model is parameter-free, whereas the second one presents the interest rate as its only dynamical parameter.

\smallskip

\begin{example}\label{ex:1}
  The first example is the Cramer-Lundberg model \cite{m2009}. In this model the surplus process is driven by premiums $G = (G_n)_{n \in \N_0}$ and claims $C = (C_n)_{n \in \N_0}$, both being represented by sequences of non-negative i.i.d. random variables. We denote the cumulative distribution functions of $G_n$ and $C_n$ as $F_G$ and $F_C$, respectively. The update equation for the surplus process has the following form:
  \begin{equation}\label{model:CL}
    Z_{n+1} = Z_n+G_n-C_n,
  \end{equation}
  where $Z_0 = z\in \mathbb R$. Since there are no parameters in the model, we define $\Theta = \{\theta\}$ to be an auxiliary singleton space. Furthermore, we chose $\Xi = \R^2$ with $\xi = (G,C)$. Then the functions $g$ and $h$ are given by $g(Z,\theta,\xi) = Z+G-C$ and $h(Z,\theta,\xi) = \theta$. Since the parameter $\theta$ plays no role in the Cramer-Lundberg model, in the following we shall notationally omit it.
\end{example}

\smallskip

\begin{example}\label{ex:2}
  As a second example, we consider a model that takes into account the effect of interest \cite{c2002,wh2008,yz2006}.
  The dynamics for the surplus is the following:
  \begin{equation}\label{model:CL-interest:Z}
    Z_{n+1} = (Z_n+G_n)(1+I_n) - C_n,
  \end{equation}
  where $Z_0 = z\in \mathbb R$ and $G$ and $C$ are as in Example \ref{ex:1}. The process $I$ represents the short-term interest rate, and is given by a sequence of non-negative random variables with the following dynamics \cite{c2002}:
  \begin{equation}\label{model:CL-interest:I}
    I_{n+1} = \alpha I_n+W_n,
  \end{equation}
  where $I_0 = i\in \RR$,
  the constant $\alpha \in [0,1)$,
  and $(W_n)_{n \in \N_0}$ is a sequence of non-negative i.i.d. random variables with cumulative distribution function denoted by $F_W$.
  In order to describe this model within the framework of \eqref{model:general},
  we select $\Theta = \R_{\geq 0}$ to be the state space for the short-term interest rate (the only parameter with dynamics in the model),
  whereas $\Xi = \R^3$ with $\xi = (G,C,W)$.
  Then $g(Z,I,\xi) = (Z+G)(1+I)-C$, and $h(Z,I,\xi) = \alpha I + W$.
\end{example}


\subsection{The ruin probability as a reachability probability}
\label{ssec:ruin-inv}

Given a risk process $X = (Z,\theta)$ as in \eqref{model:general},
the ruin probability \cite{a2000} is defined as
\begin{equation}\label{eq:ruin}
  \psi(z,\theta) = \pr_{z,\theta}\left\{\inf\limits_{n \in \N_0}Z_n < 0\right\},
\end{equation}
and corresponds to the probability that, at some moment, the capital of the insurance company becomes negative.
As mentioned in the Introduction,
this formulation can be encompassed by a class of problems in the area of formal verification of stochastic processes that are known as \emph{reachability problems} \cite{bk2008}.
This class of problems is defined as follows.
Let $A\in \b(E)$ be a given set referred to as the ``target set''.
The first hitting time of $A$ is a random variable $\tau_A:\Omega\to\N_0\cup\{\infty\}$,
defined as $\tau_A = \inf\{n \in \N_0:X_n\in A\}$.
The value function corresponding to the above reachability problem denotes the probability that the first hitting time is finite:
\begin{equation*}
  v(x;A) := \pr_x\{\tau_A<\infty\}.
\end{equation*}
Clearly, if we select the target set to be of the form $A = \R_{<0}\times \Theta$,
the following identity holds:
\begin{equation}\label{eq:ruin-reach}
  \psi(z,\theta) = v(z,\theta;\R_{<0}\times \Theta),
\end{equation}
which shows that a ruin probability is a reachability probability for the set $A$ introduced above.
Thus, the ruin problem can be tackled using techniques -- both theoretical and computational -- developed
for reachability verification \cite{APLS08b,bk2008,ta2012a-arXiv}.

Let us recall some facts about the reachability value function $v(x;A)$ (details can be found in \cite{ta2012a-arXiv}).
First of all,
the problem allows for two alternative characterizations:
either as a limit of iterations,
or as a solution of a fixpoint equation.
With focus on the former approach,
for $n\in \N_0$ let us denote by $v_n(x;A) := \pr_x(\tau_A\leq n)$ the bounded-horizon reachability value functions.\footnote{An alternative interpretation is that,
for a fixed $x$, the sequence $(v_n(x;A))_{n\in \N_0}$ is the cumulative distribution function of the random variable $\tau_A$ with respect to the measure $\pr_x$.}
We have that
\begin{equation}\label{eq:iteration-v}
  \begin{cases}
    v_{n+1}(x;A) &= 1_A(x)+\i_{A^c} v_n(x;A),
    \\
    v_0(x;A) &= 1_A(x),
  \end{cases}
\end{equation}
and the continuity of the probability measure leads to $v(x;A) = \lim_{n\to\infty}v_n(x;A)$,
where the limit point-wise on $E$ is monotonically non-decreasing.
As a result, $v(x;A)$ is the least non-negative solution of the following Bellman fixpoint equation:
\begin{equation}\label{eq:fixpoint_v}
  v(x;A) = 1_A(x)+\i_{A^c} v(x;A),
\end{equation}
that is, for any $v^*$ that is a solution of \eqref{eq:fixpoint_v} with $v^*\geq 0$,
it holds that $v(x;A)\leq v^*(x;A)$ for all $x\in E$.

\smallskip

\begin{remark}\label{rem:fixpoint-v}
  Notice that neither of the characterizations above allows expressing the reachability probability function $v(x;A)$ explicitly,
  due to the following reasons.
  First, the Bellman fixpoint equation in \eqref{eq:fixpoint_v} always admits a trivial solution,
  namely $v(x;A) = 1$ for all $x\in X$,
  which can be non-unique \cite{ta2012a-arXiv}.
  Second, the operator $\i_{A^c}$ in \eqref{eq:iteration-v} is in general not contractive,
  thus even though the convergence is monotonic, the rate of convergence of $v_n\to v$ as $n\to\infty$ is unknown.
\end{remark}

\smallskip

We have discussed that equation \eqref{eq:ruin-reach} relates the calculation of the ruin probability to that of a reachability problem.
As an immediate consequence, the recursions in \eqref{eq:iteration-v} together with equation \eqref{eq:fixpoint_v} apply to the ruin problem for any risk model given in the form \eqref{model:general}.
As examples from the literature,
the characterization of recursions for the finite-time ruin probabilities given in \cite[Lemma 2.1]{c2002}, \cite[Lemma 2.1]{cd2004}, \cite[Theorem 2.1]{wh2008} and \cite[Section 4.4]{yz2006} are special cases of \eqref{eq:iteration-v} for the corresponding risk models.
Much in the same way, the fixpoint equations presented in \cite[Lemma 2.1]{c2002}, \cite[Lemma 2.1]{cd2004}, \cite[Theorem 2.2]{wh2008} and \cite[Section 4.5]{yz2006} are special cases of \eqref{eq:fixpoint_v}.
Notice that the non-uniqueness of the solution of these fixpoint equations follows from the fact that $\psi\equiv 1$ is always a solution,
whereas the ruin probability is often not a constant function of the initial capital.
As discussed, alternative iterative approaches to compute the desired solution of these problems may run into the convergence issues described in Remark \ref{rem:fixpoint-v}.

\subsection{Solution of the reachability problem and general computation of the ruin probability}

The formal connection between ruin probability problems and probabilistic reachability ones suggests that the investigation of the latter class can yield insight on the former.
In order to tackle the technical issues discussed in Remark \ref{rem:fixpoint-v},
we leverage a new approach to study reachability problems and consequently also ruin problems.

\smallskip

With the goal of dealing with the possible non-uniqueness of the solution of the Bellman equation for the (infinite horizon) reachability problem,
we introduce a related one, known within the formal verification area: the reach-avoid problem \cite{rcsl2010}.
The reach-avoid problem can be considered as a generalization of the two-barrier ruin problem \cite[Section XI.1]{a2000}.
More precisely, for any two sets $A,B\in \b(E)$ we define the reach-avoid probability as follows:
\begin{equation*}
  w(x;A,B) := \pr_x\left\{\tau_B<\tau_{A^c},\tau_B<\infty\right\}.
\end{equation*}
As a result, $w(x;A,B)$ is the probability that,
starting from the initial condition $x\in E$,
the process will eventually hit the set $B$,
while always staying within the set $A$ beforehand.
As an example, the two-barrier ruin problem,
which evaluates the probability that the capital of an insurance company will hit a given threshold $y$ before the company goes bankrupt,
can be expressed via the following formula:
\begin{equation*}
  \phi(z,\theta,y) := w(z,\theta;\R_{\geq 0}\times \Theta,\R_{\geq y}\times \Theta),
\end{equation*}
Let us now recall how one can characterize and compute in general the value of $w(x;A,B)$.
As with the value functions for probabilistic reachability,
one can employ finite-horizon value functions for the reach-avoid problem in order to characterize the infinite-horizon value function $w$.
More precisely, let us define for any $n\in \N_0$ and sets $A,B\in \b(E)$:
\begin{equation*}
  w_n(x;A,B) := \pr_x\left\{\tau_B<\tau_{A^c},\tau_B\leq n\right\}.
\end{equation*}
It was shown in \cite{rcsl2010} that functions $w_n$ satisfy the following recursions, for all $n\in \N_0$:
\begin{equation*}
  \begin{cases}
    w_{n+1}(x;A,B) &= 1_B(x)+\i_{A\setminus B} w_n(x;A,B),
    \\
    w_0(x;A,B) &= 1_B(x).
  \end{cases}
\end{equation*}
Let us mention, that there exist well-developed numerical methods to find functions $w_n$ with any given precision based on the state space discretization \cite{aklp2010}.

Furthermore, it holds that $w(x;A,B) = \lim_{n\to\infty}w_n(x;A,B)$ \cite{rcsl2010},
where the limit is monotonically non-decreasing point-wise on $E$.
It follows that the fixpoint characterization holds as:
\begin{equation}\label{eq:fixpoint_w}
  w(x;A,B) = 1_B(x) + \i_{A\setminus B}w(x;A,B).
\end{equation}

For any $x\in E$ and $B\in \b(X)$ it holds that $v(x;B) = w(x;E,B)$,
thus the recursions and the fixpoint equation for the reachability value functions are special cases of those for the reach-avoid ones.
As a consequence,
in general the recursions and the fixpoint equation for the reach-avoid problem suffer from the issues discussed in Remark \ref{rem:fixpoint-v}.
However, we next show that,
by choosing sets $A$ and $B$ appropriately,
these issues can be mitigated:
in particular,
we attain the uniqueness of the solution of the Bellman equation and the contractivity of the operator $\i_A$,
which leads to computable solutions of the reach-avoid problem \cite{ta2012a-arXiv}.
Furthermore, we show how the solution of the reachability problem can be approximated by means of the reach-avoid one.
To achieve these goals we need the following lemma.

\smallskip

\begin{lemma}\label{lem:decomp}
  For any two sets $K,L\in \b(X)$
  \begin{itemize}
    \item[i.] the following bound holds:
      \begin{equation*}
        \left|v(x;K) - \left(1-w\left(x;K^c,L\right)\right)\right|\leq \sup_{x\in L}v(x;K);
      \end{equation*}
    \item[ii.] if $v(\cdot;K \cup L) \equiv 1$, then $w(x;K^c,L)$ is the unique solution of the corresponding version of equation \eqref{eq:fixpoint_w}.
    In particular, if there exists an integer $m\in \N_0$ for which it holds that $\delta_m:= \inf_{x\in X}v_m(x;K\cup L)>0$,
    then $v(\cdot;K \cup L) \equiv 1$, and furhtermore
      \begin{equation}
        0\leq w(x;K^c,L) - w_n(x;K^c,L)\leq \frac{m}{\delta_m}(1-\delta_m)^{\left\lfloor\frac{n}{m}\right\rfloor}
      \end{equation}
      for any $x\in X$ and for any $n \in \N_0$.
  \end{itemize}
\end{lemma}

\begin{proof}
  Part [i.] follows directly from \cite[Lemma 2]{ta2012a-arXiv}, whereas part [ii.] can be directly shown by \cite[Proposition 2, Theorem 1]{ta2012a-arXiv}.
\end{proof}

\smallskip

Let us discuss how Lemma \ref{lem:decomp} leads to attain the goals described above.
Part [i.] shows that the reachability function $v(\cdot;K)$ can be approximated using the reach-avoid function $w(\cdot;K^c,L)$,
which leads to the computation of the latter.
Part [ii.] in turn provides conditions to compute the infinite horizon reach-avoid function $w(x;K^c,L)$ using finite-horizon ones $w_n(x;K^c,L)$.
Note that the latter finite-horizon functions can be computed with available schemes \cite{aklp2010}.

Before we apply Lemma \ref{lem:decomp} to the case of the ruin probability problem,
let us elucidate the direction we are going to pursue.
With focus on $v(x;A)$,
a given set $A$ allows for the following dichotomy depending on the quantity $\alpha := \inf_{x\in X}v(x;A)$:

\begin{enumerate}
  \item If $\alpha > 0$, then there exists $m\in \N_0$ such that $\inf_{x\in X}v_m(x;A) > 0.$
  As a result, if in [ii.] of Lemma \ref{lem:decomp} we take $K = \emptyset$ and $L = A$, we get that $v\equiv 1$, hence the problem is solved.
  \item If $\alpha = 0$, then we are able to find the set $B\subset A^c$ small enough such that $\sup_{x\in B}v(x;A)$ is less than a required precision level. If at the same time the set $B$ is big enough such that $\delta_m$ as per [ii.] of Lemma \ref{lem:decomp} is positive, then
      \begin{equation}\label{eq:bounds-full}
        \left|v(x;A) - \left(1-w_n\left(x;A^c,B\right)\right)\right|\leq \sup_{x\in B}v(x;A) + \frac{m}{\delta_m}(1-\delta_m)^{\left\lfloor\frac{n}{m}\right\rfloor}.
      \end{equation}
\end{enumerate}
As it has been mentioned above, functions $w_n$ can be computed numerically with any given precision, so \eqref{eq:bounds-full} provides a useful way of approximating the reachability probability $v$.
The crucial step here is, however, the choice of $B$ and the construction of bounds on $\sup_{x\in B}v(x;A)$.
While there is no general procedure for that, let us show how the described idea applies to the case of the ruin probability problem.

\smallskip

We start with the choice of the set $B$.
Recall that in order to connect ruin and reachability problems by \eqref{eq:ruin-reach} we pick a set $A_0 := \R_{\geq 0}\times \Theta$.
We also denote by $B_y = \R_{> y}\times \Theta$ the candidate for the choice of the set $B$, for $y>0$. Thus
\begin{equation*}
  w(z,\theta; A_0,B_y) = \phi(z,\theta,y)
\end{equation*}
corresponds to the two-barrier ruin probability, with barriers chosen to be equal to $0$ and $y$.
Note that from Lemma \ref{lem:decomp} it follows that in order for $\phi$ to be the unique solution of the corresponding fixpoint equation
\begin{equation}\label{eq:fixpoint_phi}
  \phi(z,\theta,y) = 1_{(y,\infty)}(z)+1_{[0,y]}(z)\cdot\int\limits_{\R\times \Theta} \phi(z',\theta',y)P((z,\theta),\mathrm dz'\times \mathrm d\theta')
\end{equation}
it is sufficient that $v(\cdot;A_0\cup B_y)\equiv 1$.
With focus on the latter condition,
let us remark that for most of the models of risk processes an even stronger condition holds true,
namely that $v(\cdot;B_y)\equiv 1$.
This is always satisfied whenever the expected value of the increment of the risk process is greater than a positive constant.
This condition corresponds to the known Net Profit Condition (NPC) \cite{m2009},
which assures that the insurance company is profitable in the long run (examples of NPC are given in Sections \ref{ssec:CL}, \ref{ssec:CL-I}).

Clearly, for any fixed $(z,\theta) \in \R\times\Theta$,
we have that $\lim_{y\to\infty}\phi(z,\theta,y) = 1-\psi(z,\theta)$.
This fact has been used in \cite{a2000} over a few specific instances where $\phi$ admits a closed form solution,
so that the exact solution for $\psi$ has been obtained by taking the limit $y\to\infty$.
From the new perspective offered by Lemma \ref{lem:decomp},
since function $\phi$ can now be computed via $w$,
the goal is to find a $y$ such that $w$ and $1-\psi$ are close enough.
In other words,
the discussion above suggests to approximate the ruin probability $\psi$ via computable two-barrier ruin probabilities $\phi$,
which allows to formulate a version of Lemma \ref{lem:decomp} for ruin probabilities.
In order to achieve practically useful results,
let us raise the following assumption on the model of the risk process:

\smallskip

\begin{assumption}\label{ass:1}
  Suppose that in \eqref{model:general}, the map $h$ does not depend on $z$: $h(z,\theta,\xi) = h(\theta,\xi)$.
  Moreover, assume that the map $g$ satisfies $g(z',\theta,\xi)\leq g(z'',\theta,\xi)$ whenever $z'\leq z''$.
\end{assumption}

\smallskip

Assumption \ref{ass:1} can be interpreted as follows.
Firstly, the evolution of the parameters of the model,
described by function $h$, is assumed to be independent of the capital of the company.
Secondly, for a fixed value of the parameters and for a fixed realization of the noise,
the higher is the current value of the capital, the higher will be the next value in time.
Such assumptions are satisfied by a wide class of models,
in particular by those discussed in Examples \ref{ex:1} and \ref{ex:2}.

\smallskip

\begin{theorem}\label{thm:main}
  Let Assumption \ref{ass:1} hold true.
  Denote by $\psi^*(y):= \sup_{\theta\in \Theta}\psi(y,\theta)$.
  Then for any $(z,\theta) \in \R\times \Theta$ and for all $y>0$,
  it holds that
  \begin{equation}\label{eq:bounds-main}
    |\psi(z,\theta) - (1 - \phi(z,\theta,y))|\leq \psi^*(y).
  \end{equation}
\end{theorem}

\begin{proof}
  From Lemma \ref{lem:decomp} it immediately follows that
  \begin{equation*}
    |\psi(z,\theta) - (1 - \phi(z,\theta,y))| \leq \sup_{\theta\in \Theta}\sup_{z>y}\psi(z,\theta),
  \end{equation*}
  thus it is sufficient to show that $\sup_{z>y}\psi(z,\theta) = \psi(y,\theta)$,
  for any fixed $\theta\in \Theta$.
  In order to show the latter equality,
  it is sufficient to prove that $\psi(\tz,\theta)\leq\psi(z,\theta)$ if $z\leq \tz$ for all $\theta\in \Theta$:
  this is done by coupling techniques \cite{l1992}.
  Consider two arbitrary points $z,\tz\in \R$ such that $z\leq \tz$,
  and a $\theta_0\in \Theta$.
  Let us denote by $(Z,\theta)$ and $(\tilde Z,\theta)$ two processes defined by \eqref{model:general} and starting respectively from $(z,\theta_0)$ and $(\tz,\theta_0)$.    Under the assumptions in the statement of the proposition it holds that $Z_n(\omega)\leq \tilde Z_n(\omega)$ for any $n\in \N_0$ and any fixed $\omega\in \Omega$ so the following inclusion holds true:
  \begin{equation}\label{eq:thm.main}
    \left\{\omega\in \Omega:\inf\limits_{n \in \N_0}\tilde Z_n(\omega)<0\right\}\subseteq\left\{\omega\in \Omega:\inf\limits_{n \in \N_0}Z_n(\omega)<0\right\}.
  \end{equation}
  Finally, since $\psi(\tz,\theta)$ is the probability of the left-hand side of \eqref{eq:thm.main},
  and $\psi(z,\theta)$ is that of the right-hand side,
  it follows that $\psi(\tz,\theta)\leq \psi(z,\theta)$.
\end{proof}

\smallskip

Let us discuss the implications of Theorem \ref{thm:main}.
Clearly, to apply the result one needs the information on the upper bounds on $\psi^*$.
The literature on the ruin probability problem shows a remarkable interest in studying the asymptotic behavior of the quantity $\psi^*(y)$ for large values of $y$,
since it provides an upper-bound on the ruin probability whenever the capital is large enough.
However, up to our knowledge such bounds are rarely tight \cite{fkz2011,k2011},
hence they may not be practically relevant whenever the capital is not sufficiently large,
and in particular not useful to estimate $\psi(0,\theta)$.
Notice that although $\psi(z,\theta) = 1$ for $z<0$, in general
$\psi(0,\theta)$ can still be much smaller than $1$ (cfr. Section \ref{sec:cs}).
Such a value may be of interest since it describes the ruin probability for negligibly small initial capitals.

Let us then emphasize the relationship between our results and those on the asymptotic behavior of ruin probabilities.
Although Theorem \ref{thm:main} requires the knowledge of an explicit upper bound on $\psi^*$,
it further shows that such a bound is not only useful if we are interested in $\psi(y,\theta)$ for large values of $y$,
but it can also be employed as an approximation bound for the solution of the ruin problem as in the right-hand side of \eqref{eq:bounds-main}.
As we mentioned above, the quantity $1-\phi(z,\theta,y)$ can be computed with any given precision,
and it differs from the desired value $\psi(z,\theta)$ by the quantity $\psi^*(y)$, which is usually small if $y$ is large.
As a result, Theorem \ref{thm:main} gives a new perspective on the asymptotic bound based on $\psi^*$.

Let us stress that in the literature bounds on $\psi^*(y)$ are not known for general Markov models of the risk process, even for large values of $y$.
In particular, such bounds are not known even for the basic Cramer-Lundberg model (given in Example \ref{ex:1}) in the case when the second moment of the claim variable is infinite \cite{k2011}.
In the next section we recapitulate both classic and recent results for the Cramer-Lundberg model and benchmark them with the outcome of the proposed new technique,
which is thereafter also applied to the model in Example \ref{ex:2}.
This discussion is further supported with computational examples given in Section \ref{sec:cs}.

\subsection{Asymptotic behavior of the ruin probability for the Cramer-Lundberg model}
\label{ssec:CL}

Let us recall that the Cramer-Lundberg model is given by a random walk on a real line:
\begin{equation*}
  Z_{n+1} = Z_n+\eta_n,
\end{equation*}
with $\theta_n\equiv \mathrm{const}$, and where $\eta_n = G_n - C_n$ is a sequence of iid random variables and $G$, $C$ are as in Example \ref{ex:1}.
The ruin probability over such model under the NPC \cite{m2009}
\begin{equation}\label{eq:npc-lun}
  a:=\Ex\eta_0 = \Ex G_0-\Ex C_0>0
\end{equation}
can thus be related to the tail probability of the maximum of a random walk with a negative drift,
which in this particular case is given by $-Z$ \cite{m2009}.
The distribution of the maximum of a random walk is in itself an important problem (see e.g. \cite[Chapter 5]{fkz2011}),
and further corresponds to the tail probability of the equilibrium waiting time in a G/G/1 queue \cite{c2002}.
The explicit expression of this quantity has only been found in a limited number of cases,
while most of the results are asymptotic \cite{bb2008,k1997}.
Recall that we are interested not only in the asymptotic behavior of $\psi(z)$ for $z\to\infty$,\footnote{Since $\theta$ in the model is a singleton point, we notationally omit the dependence of the ruin probability $\psi$ on $\theta$. Due to the same reason, it holds that $\psi^* \equiv \psi$.} but specifically in the upper bounds on $\psi$ that are needed in \eqref{eq:bounds-main}.

Let us mention the main results in the literature and the challenges for this problem (more details can be found in \cite{bb2008,fkz2011,k1997,m2009}).
The derivation of upper bounds on $\psi(z)$, for large $z$,
strongly depends on the properties of the moment generating function for the random variable $(-\eta_0)$,
which is given by $m(t) = \Ex \mathrm e^{-t\eta_0}$,
and in particular on the quantity $\lambda = \sup\{t\geq 0:m(t)\leq 1\}$.
In the Cramer case \cite{k1997} characterized by $\lambda>0$ and $m(\lambda) = 1$,
the quantity $\lambda$ is known as the \textit{adjustment} or \textit{Lundberg coefficient},
and it holds that
\begin{equation}\label{eq:bounds-lund}
  \psi^*(z) = \psi(z)\leq \mathrm e^{-\lambda z},
\end{equation}
for all $z\geq 0$. On the other hand, the condition $\lambda>0$ implies that $m(t)$ is finite for some positive $t$, which is not satisfied by the class of \textit{heavy-tailed} distributions --- among these are the log-normal, Pareto and Levy distributions, for which $m(t)$ does not exist for any $t>0$ and thus $\lambda = 0$.
These distributions are commonly used in both queuing theory and risk theory, in the latter case to represent large claim sizes. Although in the case of heavy-tailed distributions the asymptotic behavior of $\psi$ is well studied, non-trivial explicit upper bounds for general distributions are less common in the literature and often conservative \cite{k2011}.
Fortunately, the conservatism of bounds on $\psi^*$ is not detrimental to the efficient application of \eqref{eq:bounds-main} --
the only property that matters is that $\lim_{y\to\infty}\psi^*(y) = 0$.
We thus adapt the bounds from \cite{k2011} to the case of the Cramer-Lundberg model.

\smallskip

\begin{theorem}\label{thm:korsh}
  For the Cramer-Lundberg model \eqref{model:CL} assume the following:
  \begin{enumerate}
    \item NPC, as in \eqref{eq:npc-lun};
    \item $F_C(x)<1$, for all $x>0$;
    \item $\Ex |C_0 - G_0|^\gamma<\infty$ for some $\gamma\geq 2$.
  \end{enumerate}
  The the following bounds hold true for all $z,y>0$:
  \begin{equation}\label{eq:bounds-korsh}
    |\psi(z) - (1-\phi(z,y))|\leq c \cdot y^{1-\gamma}.
  \end{equation}
  Here $c = \frac{3\max(c_1,c_2)}{a\gamma}+\frac12s_3^{\gamma-1}$,
  where $s_3 = \max(s_1,s_2)$,
  $s_1$ is any real number such that $\Ex \min(G_0 - C_0,s_1)\geq \frac23a$,
  and $s_2 = 2^{\gamma-1}\frac{\gamma-1}{a}\Ex(G_0-C_0)^2$.
  Also, $c_1$ and $c_2$ are given by
  \begin{align*}
    c_1 &= \gamma(\gamma-1)2^{\gamma-3}\Ex C_0^{\gamma-2}(G_0-C_0)^2,
    \\
    c_2 &= \gamma \Ex(s_3+C_0)^{\gamma-1}C_0.
  \end{align*}
\end{theorem}

\begin{proof}
 Define the random walk $S = (S_n)_{n \in \N_0}$  on the same probability space of $Z$ by
  \begin{equation}
    S_{n+1} = S_n - G_n+C_n, \quad S_0 = 0.
  \end{equation}
  Further let $S^*:=\sup_{n \in \N_0}S_n$ be the maximum of the random walk $S$.
  It immediately follows from its definition that $\pr(S^*>y) = \psi(y)$, for any $y\geq 0$.
  On the other hand, it has been shown in \cite{k2011} that under the assumptions of the theorem it holds that $\Ex (S^*)^{\gamma-1}<c$. Due to the fact that $S^*$ is non-negative, by the Markov inequality we obtain that $\pr(S^*>y)\leq c\cdot y^{1-\gamma}$.
  The rest follows from \eqref{eq:bounds-main}.
\end{proof}

\smallskip

Some remarks are needed for Theorem \ref{thm:korsh}.
With focus on the assumptions in the statement,
the NPC is one of the most commonly-employed conditions for risk models,
and it is equivalent to assuming that the increment of the capital of the insurance company during one time period is on average positive.
If instead for the Cramer-Lundberg model it holds that $\Ex G_0 - \Ex C_0<0$,
then $\psi(z)=1$ for all values of the initial capital $z$.
Further, the condition $F_C(x)<1$ for all $x\geq 0$ means that the distribution of $C_0$ has an unbounded support on $\R_{\geq 0}$.
Whenever this is not satisfied, it holds that $C_0$ is bounded almost surely,
which leads to tighter bounds in \eqref{eq:bounds-lund}.
Let us also mention that under the assumptions in Theorem \ref{thm:korsh},
the quantities $c_1,c_2$ and $s_2$ are obviously finite.
Taking into account the fact that $a>0$, the finiteness of $s_1$ is also clear.

The last condition in the statement of Theorem \ref{thm:korsh},
namely $\Ex|C_0-G_0|^\gamma<\infty$ for some $\gamma\geq 2$, is more interesting.
First of all, notice that this condition can be relaxed as the following:
$\Ex C_0^\gamma<\infty$, for some $\gamma\geq 2$.
Indeed, if $a=\Ex G_0 - \Ex C_0>0$ then there always exists a $k\in \R_{\geq 0}$ such that
\begin{equation}\label{eq:trunc}
  \Ex \min(k,G_0) - \Ex C_0>0.
\end{equation}
For such a $k$ let us define $\hat G_0:=\min(k,G_0)$ to be the new variable for the premiums and by $\hat Z$ the corresponding risk process.
Clearly, if $\Ex C_0^\gamma<\infty$ then $\Ex |C_0 - \hat G_0|^\gamma<\infty$ and thus $\hat Z$ satisfies the last condition of Theorem \ref{thm:korsh}.
On the other hand, by coupling $Z$ and $\hat Z$ on the same probability space we obtain that $\hat Z\leq Z$ a.s. and thus $\psi(z)\leq \hat\psi(z)$ for all $z\geq 0$,
where $\hat\psi$ is the ruin probability for the risk process $\hat Z$.
Thus the bound in \eqref{eq:bounds-korsh} still holds if one relaxes the assumption $\Ex|C_0-G_0|^\gamma<\infty$ for some $\gamma\geq 2$,
to $\Ex C_0^\gamma<\infty$ for some $\gamma\geq 2$.
This can be considered to be the main requirement of Theorem \ref{thm:korsh}:
notice that the majority of the distributions for the claim size considered in the literature \cite{a2000,m2009},
even those extremely heavy-tailed,
has a final second moment and thus satisfies this assumption.

\smallskip

In order to further elucidate the overall technique presented in this work,
let us formulate the proposed solution of the ruin problem on the Cramer-Lundberg model.
Let the claim size $C_0$ and the premiums $G_0$ be given random variables, and assume that $\Ex C_0^2<\infty$ and $a = \Ex(G_0 - C_0) >0$.

The procedure goes as follows:
\begin{enumerate}
  \item[1.] pick up a precision level $\ve>0$
  \item[2a.] if $C_0$ has a bounded support, find the Lundberg coefficient $\lambda$ for $\eta_0 = G_0 - C_0$, and define quantity $y = -\frac1\lambda\log\frac12\ve$
  \item[2b.] if $C_0$ has an unbounded support, compute $\Ex G_0^2$.
  If $\Ex G_0^2 = \infty$ follow the truncation procedure and find $k>0$ satisfying \eqref{eq:trunc}.
  Select $y = \frac{2c}{\ve}$, where $c$ is computed as in Theorem \ref{thm:korsh} either for $G_0$ or for $\min(k,G_0)$
  \item[3.] depending on whether support of $C_0$ is bounded or not, take $y$ as defined either in [2a.] or in [2b.];
  find $\phi(z,y)$ with precision $\frac12\ve$
  \item[4.] define $\tilde\psi(z):=1-\phi(z,y)$.
  By Theorem \ref{thm:main}, it holds that $\|\psi - \tilde\psi\|\leq\ve$
\end{enumerate}

\medskip

It is also worth clarifying the way the two-barrier ruin probability is computed.
As above, let $\eta_0 = G_0 - C_0$ and let $F_\eta$ be its cumulative distribution function. Then
\begin{equation*}
  \phi(z,y) = 1_{(y,\infty)}(z) + 1_{[0,y]}(z)h_y(z),
\end{equation*}
where from \eqref{eq:fixpoint_phi} it follows that $h_y:[0,y]\to\R$ is the solution of the fixpoint equation
\begin{equation*}
  h(z,y) = F_\eta(y-z)+\int_0^t h_y(t)\mathrm dF_\eta(t-z),
\end{equation*}
which in case $F_\eta$ admits a density $f_\eta$ is a Fredholm integral equation of the second kind on a compact interval \cite{a1997} as
\begin{equation}\label{eq:fred-phi}
  h(z,y) = F_\eta(y-z)+\int_0^t h_y(t)f_\eta(t-z)\mathrm dt.
\end{equation}
Numerical methods can be applied to find the solution of this equation.
Section \ref{sec:cs} presents an example of this computational procedure.

\subsection{Asymptotic behavior of the ruin probability for a model with interest rates}
\label{ssec:CL-I}

Let us now tailor the main result in the previous section to the model given in Example \ref{ex:2}.
The presented bounds are derived directly through the corresponding Cramer-Lundberg model and are thus subject to similar conditions.
More precisely, let $(Z,I)$ be a risk process given by \eqref{model:CL-interest:Z}-\eqref{model:CL-interest:I}.
The NPC for this model is $\Ex G_0 - \Ex C_0>0$ \cite{c2002},
which is similar to that for the Cramer-Lundberg model.
The work in \cite{c2002} also shows that $\psi(z,i) \leq \hat\psi(z)$ for all $z,i\geq 0$,
where the latter ruin probability is defined over the Cramer-Lundberg risk process $\hat Z$ corresponding to $(Z,I)$ and characterized by
\begin{equation*}
  \hat Z_{n+1} = \hat Z_n+G_n-C_n,\quad \hat Z_0 = z.
\end{equation*}
Thus, in order to find upper-bounds on $\psi$,
we can leverage the ones on $\hat \psi$,
hence the procedure given in Section \ref{ssec:CL} directly applies.
The final bounds have the form $\|\psi - \tilde\psi\|\leq \ve$, where
\begin{equation}\label{eq:CL-I.approx}
  \tilde\psi(z,i) = 1 - \phi(z,i,y(\ve)),
\end{equation}
and where $y(\ve)$ is chosen based on the bounds for $\hat\psi$, either as in \eqref{eq:bounds-lund} or as in \eqref{eq:bounds-korsh}.

Although the approximation in \eqref{eq:CL-I.approx} can be made as precise as possible,
it involves the computation of the two-barrier probability $\phi$ over a domain that is now unbounded:
this is because of the interest variable that takes values over $\R_{\geq 0}$.
Since numerical methods for the solution of integral equations over unbounded sets are available only in a limited number of cases,
such approximation may not be useful in practice.
In order to cope with this issue, instead of approximating $\psi$ by the two-barrier ruin probability,
one can approximate it with the solution of another reach-avoid problem where the target set is defined as
\begin{equation*}
  B_{y,j} = \{(z,i):z> y\text{ or }i> j\}
\end{equation*}
and where the set of allowed states is as before $A_0 = \R_{\geq 0}\times \R_{\geq 0}$.
The fixpoint equation for the corresponding value function $w(z,i;A_0,B_{y,j})$ needs to be solved only over the compact set $A_0\setminus B_{y,j} = [0,y]\times [0,j]$, and dedicated numerical methods \cite{aklp2010} can be applied to find its solution with any given precision.
The result of Theorem \ref{thm:main}
can be adapted to this case as
\begin{equation}\label{eq:bounds-CL-I}
  \left|\psi(z,i) -1+w(z,i;A_0,B_{y,j})\right|\leq \max\left(\psi(y,0),\psi(0,j)\right),
\end{equation}
which leads to seeking for bounds on the right-hand side of \eqref{eq:bounds-CL-I}.
As mentioned above,
$\psi(y,0) \leq \hat\psi(y)$ for all $y\geq 0$,
and the latter probability can be bounded using results from Section \ref{ssec:CL}.
Thus only the quantity $\psi(0,j)$ is left to be studied.

In order to reach any possible precision,
we have to show that $\lim_{j\to\infty}\psi(0,j) = 0$.
Clearly, this is not the case if $F_G(0)>0$, since by considering the ruin event at the first step we have
\begin{equation}
  \psi(0,j)\geq \pr\{G_0 =0\}\pr\{C_0<0\} = F_G(0)(1 - F_C(0))
\end{equation}
regardless of the value of $j$.
To avoid this, we assume that $F_G(0) = 0$. Let us now fix an arbitrary $y>0$. Since $\psi$ is a reachability probability it admits \eqref{eq:fixpoint_v}, so:
\begin{align*}
\psi(0,j) &= P \psi(0,j) = \Ex_{(0,j)}[\psi(Z_1,I_1)]
\\
&\leq \sup\limits_{z<y}\psi^*(z)\cdot\pr_{(0,j)}\{Z_1<y\}+\sup\limits_{z\geq y}\psi^*(z)\cdot\pr_{(0,j)}\{Z_1\geq y\}
\\
&\leq \pr_{(0,j)}\{Z_1<y\}+\psi^*(y).
\end{align*}
Recall that $\psi^*(y) = \sup\limits_{i\geq 0}\psi(y,i) = \psi(y,0)$. Furthermore, for $j>0$ and all $\beta\in (0,1)$:
\begin{align*}
  \pr_{(0,j)}\{Z_1<y\} &= \pr\{G_0(1+j)-C_0<y\}
  \\
  &\leq \pr\left\{G_0-C_0< -j^\beta\Ex|G_0-C_0|\right\}+\pr\left\{j\cdot G_0<y+j^\beta\Ex|G_0-C_0|\right\}
  \\
  &\leq \pr\left\{|G_0-C_0|> j^\beta\Ex|G_0-C_0|\right\}+F_G\left(\frac{y+j^\beta\Ex|G_0-C_0|}j\right)
  \\
  &\leq \frac1{j^{\beta}}+F_G\left(\frac{y}{j}+\frac{\Ex|G_0-C_0|}{j^{1-\beta}}\right).
\end{align*}
As a result, for any $\beta\in (0,1)$ and $(z,i)\in \R\times \RR$, it holds that
\begin{equation}\label{eq:model-Cl-interest-bounds-3}
\left|\psi(z,i) -1+w(z,i;A_0,B_{y,j})\right|\leq \hat\psi(y)+\frac{1}{j^\beta}+F_G\left(\frac{y}{j}+\frac{\Ex|G_0-C_0|}{j^{1-\beta}}\right).
\end{equation}
Since $F_G(0) = 0$ and $F_G$ is right-continuous we have
\begin{equation}
  \lim\limits_{j\to\infty}\left(\frac1{j^\beta}+F_G\left(\frac{y}{j}+\frac{\Ex|G_0-C_0|}{j^{1-\beta}}\right)\right) = 0,
\end{equation}
thus the bound in \eqref{eq:model-Cl-interest-bounds-3} is consistent.

\subsection{Discussion}

So far we have developed techniques to deal with the ruin problem over discrete-time Markov models for risk processes.
More precisely, under Assumption \ref{ass:1} the result of Theorem \ref{thm:main} allows approximating the ruin probability by means of a computable two-barrier ruin probability with a precise bound on the approximation error.
Such bound is in turn related to the value of the ruin probability for large values of the initial capital,
and thus can often be made as small as needed:
more detailed examples have been provided in Sections \ref{ssec:CL} and \ref{ssec:CL-I}.

In the case of a general discrete-time Markov model not satisfying Assumption \ref{ass:1},
we have still been able to provide the recursions and fixpoint equation for the finite-horizon and infinite-horizon ruin probabilities, respectively.
Similar results have been also obtained for the two-barrier ruin problem.
Recall that we have extensively used the methods developed for the reachability and reach-avoid problems,
which are in particular crucial in the proof of Lemma \ref{lem:decomp} and hence in that of Theorem \ref{thm:main}.
It is further worth mentioning that the reachability problem can be used to study more complicated events than the reachability itself.
Namely, a wide range of events of interest, such as
\emph{``the capital of the company always stays positive, and is above the threshold value $T>0$ at given moment of time,''}
or such as \emph{``once the capital reaches the threshold value $T>0$, it stays positive for at least the next $N$ steps,''}
can be directly expressed by means of linear temporal logics -- see e.g. \cite[Chapter 4]{bk2008} or \cite{ta2013}.
Interestingly, finding the probability of such events can be further recast as a basic reachability problem over a slightly modified model \cite[Theorem 5]{tmka2013},
which is one of the celebrated results of an approach in formal verification called model-checking \cite{bk2008}.
As a result, the properties obtained for the reachability problem apply to a much richer class of events,
characterizing the desired recursions on the finite-time horizon and the fixpoint equations on the infinite-time horizon,
together with appropriate approximation methods for the latter case.
Arguably, such techniques allow for the risk analysis of more complicated properties in insurance science,
rather than the basic yet fundamental ruin property.
In addition, similar methods can be also applied over models of risk processes that require decision making:
e.g. investing in risky assets,
or re-insuring claims \cite{dr2009}.
Indeed, \cite[Corollary 3]{tmka2013} provided a finite-horizon recursion for the maximal and the minimal reachability,
whereas \cite[Theorems 2, 3]{tmka2013} characterized the fixpoint equations.
It is likely that result similar to Lemma \ref{lem:decomp} can be also obtained in such setting,
which would allow applying techniques we developed here to decision-dependent models of risk processes.
These developments however go beyond the scope of the current contribution.

\section{Case studies}\label{sec:cs}

This section introduces two case studies that focus on the computation of the quantities of interest.

\subsection{Case study: exponential tails}\label{ssec:CS-ruin}
We consider the study of the ruin probability for the Cramer-Lundberg model worked out by Yang in \cite{y1999}, where the income is kept constant ($G_n\equiv G_0 = 1.3035, \forall n$) and the claims are distributed according to the Generalized Inverse Gaussian law with density
\begin{equation*}
  f_C(x) = 1_{\RR}(x)\frac{1}{2 k}x^{-2}\exp\left(-x-\frac1x\right),
\end{equation*}
where $k\approx 0.139866$ is a normalizing constant. For this setup it was shown in \cite{y1999} that
\begin{equation}\label{eq:bounds-yang}
  \psi(z)\leq (1+0.1 z)^{-0.1}\mathrm e^{-z}.
\end{equation}
Since this bound is not tight for small values of $z$, we apply our results to find the value $\psi(z)$ with a given precision $\ve = 0.011$ for any $z\geq 0$. The approach is based on the bound provided in Theorem \ref{thm:main}, so we first find $y$ such that $\psi(y) \leq \ve$. It follows from \eqref{eq:bounds-yang} that if $y = 4.5$, then $\psi(y)\leq 0.0107$. We thus approximate $\psi$ via the function $\tilde\psi$ expressed via the two-barrier ruin probability $\tilde\psi(z) = 1 - \phi(z,4.5)$, so that
\begin{equation*}
  \|\psi - \tilde\psi\|\leq 0.0107.
\end{equation*}

We are left with computing the value of $\phi$ on the interval $[0,4.5]$, for which we employ equation \eqref{eq:fixpoint_phi}, which in this case takes the form
\begin{equation}\label{eq:cs-yang-fred}
  \phi(z,4.5) = F_C(z - 3.1965) + \int\limits_0^{4.5} \phi(t,4.5)f_C(z+1.3035 - t)\mathrm dt,
\end{equation}
for all $z\in [0,4.5]$.
Note that \eqref{eq:cs-yang-fred} is a Fredholm integral equation of the second kind.
Since the conditions of Lemma \ref{lem:decomp} are satisfied, this equation admits a unique solution -- alternatively, considering an operator $\l$, acting on $\B\left([0,4.5]\right)$ as
\begin{equation*}
  \l g(z) = \int\limits_0^{4.5} g(t)f_c(z+1.3035-t)\,ds,
\end{equation*}
its norm is equal to $\|\l\|\leq 0.9989 <1$, which again shows by the Contraction Mapping Theorem \cite[Proposition A.1]{hl1989} that \eqref{eq:cs-yang-fred} has a unique solution. The \texttt{FIE} toolbox \cite{as2008} is used to numerically solve the integral equation with a selected error $\leq 10^{-5}$. The result is given in Figure \ref{fig:cs1} in a blue, solid line. In order to validate the obtained bounds (depending on $\ve$ -- dashed green line in Figure \ref{fig:cs1}-(a)), Monte-Carlo simulations with $2000$ initial seeds for each cyan point in Figure \ref{fig:cs1}-(a), and each run over $2000$ iterations, were used. The outcomes are close to the solution of \eqref{eq:cs-yang-fred} obtained by integration and further within the $\ve$-upper bound, which is much less conservative than the original bound \eqref{eq:bounds-yang} of Figure \ref{fig:cs1}-(b).

\smallskip

\begin{figure}[ht]
\centering
\subfigure[Solution with an $\ve$-upper bound]{\includegraphics[keepaspectratio=true,width=7.5cm]{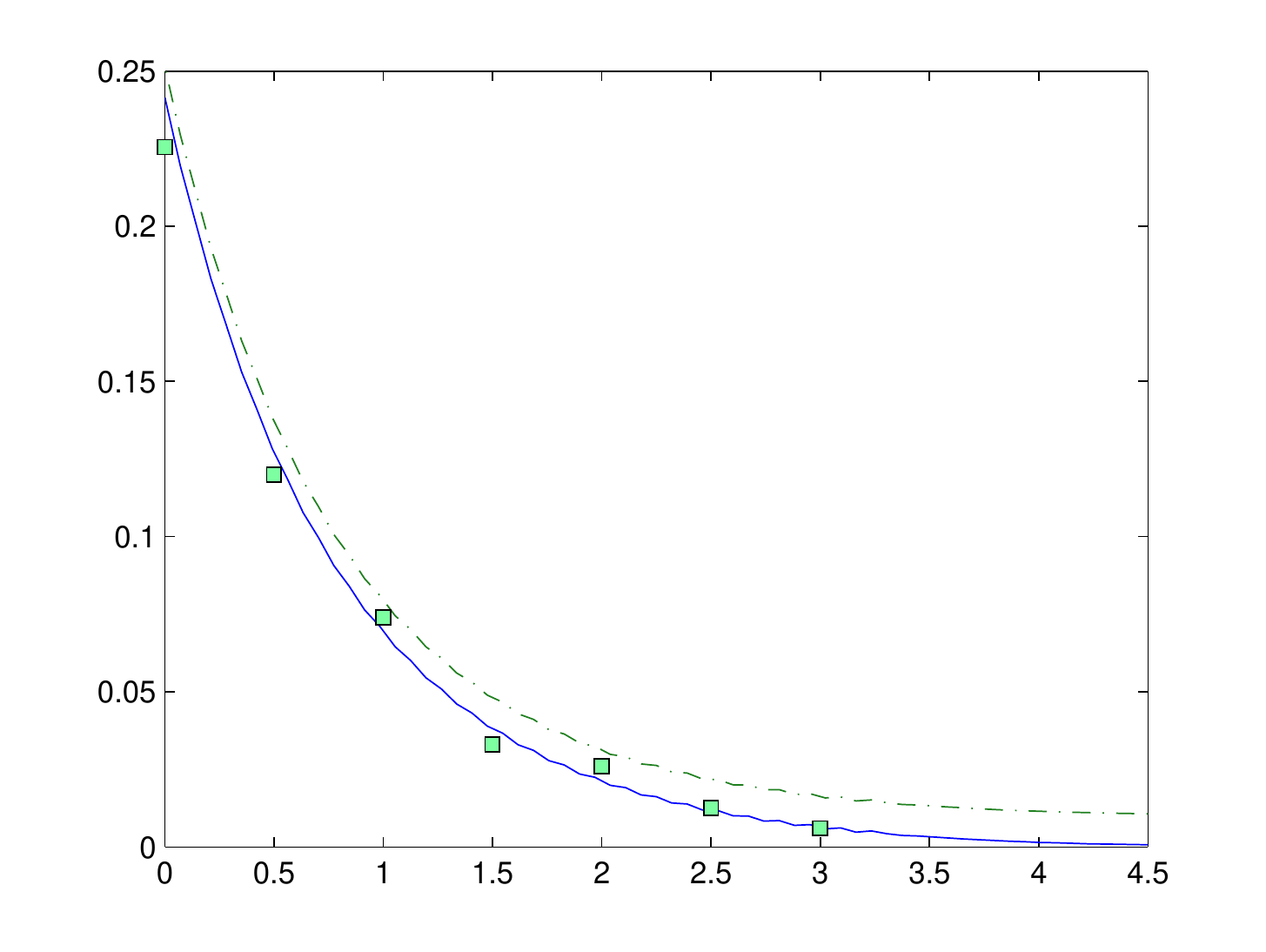}}
\subfigure[Solution with an upper bound from \cite{y1999}]{\includegraphics[keepaspectratio=true,width=7.5cm]{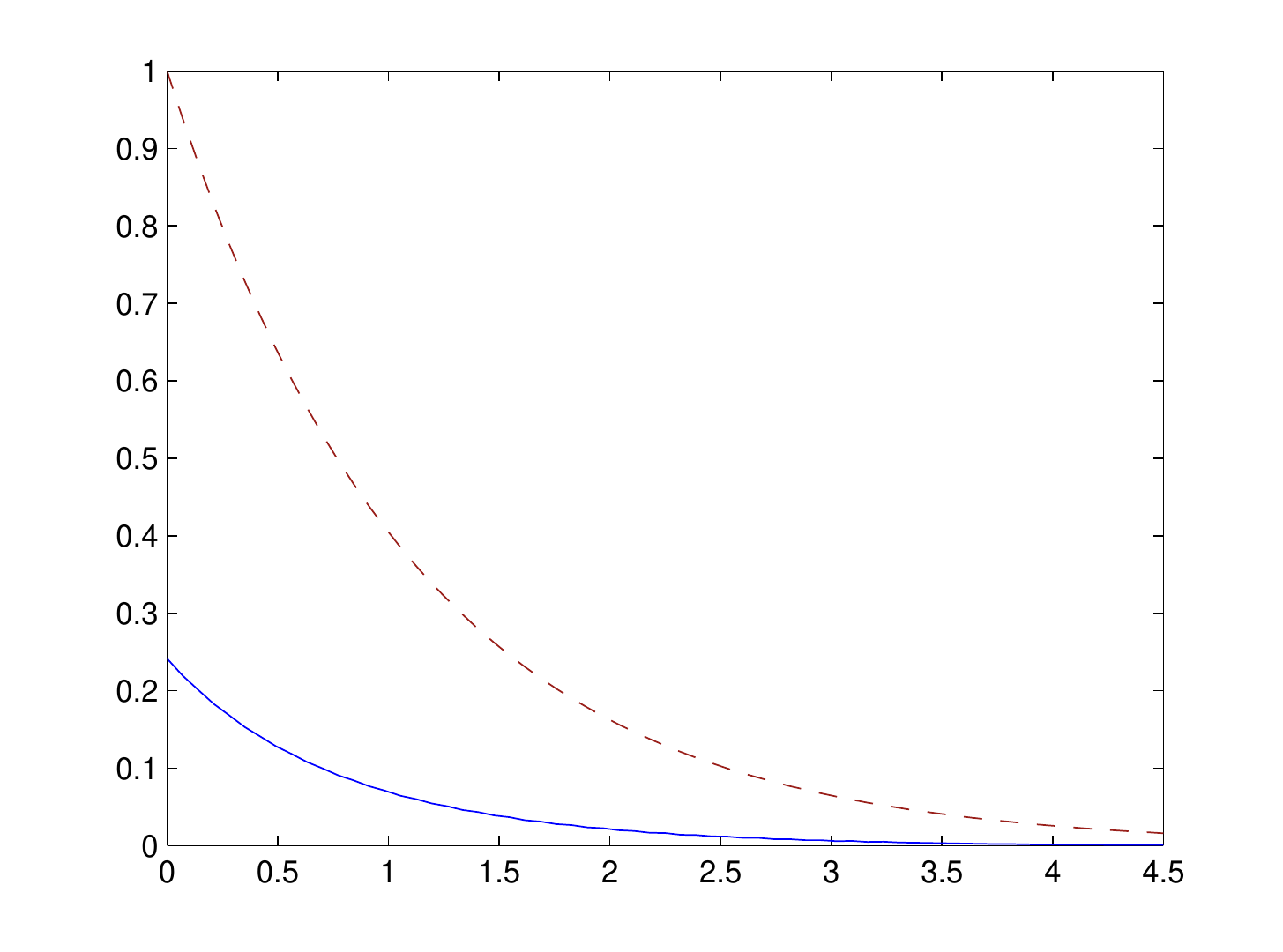}}
\caption{Solution of the ruin probability problem of Section \ref{ssec:CS-ruin} (no interest effect included).
The blue (continuous) line represents $\tilde\psi$,
the green (dashed) line in (a) provides an $\ve$-upper bound for the ruin probability $1-\phi$,
the cyan boxes correspond to results of Monte-Carlo simulations of $\psi$.
In (b), the brown (dashed) line shows the upper bound obtained in \cite{y1999}.}
\label{fig:cs1}
\end{figure}

We have also considered the model given by \eqref{model:CL-interest:Z} with the same fixed income and distribution of claims.
As in \cite{wh2008}, the interest rate is chosen to be i.i.d. $I_0 = 0.01b_0$,
where $b_0$ follows the binomial distribution $\mathrm B(10,1/2)$,
so that $\Ex I_0 = 0.05$.
Clearly for this case the bound \eqref{eq:bounds-yang} also holds,
thus we can again solve the problem on the interval $[0,4.5]$ to obtain the solution with an $\ve$-precision,
$\ve \leq 0.011$, as displayed in Figure \ref{fig:cs1-interest}.
The validation has also been performed by Monte-Carlo simulations,
with the same number of seeds and iterations as above -- the results are displayed in Figure \ref{fig:cs1-interest}-a.

\begin{figure}[ht]
\centering
\subfigure[Solution with an $\ve$-upper bound]{\includegraphics[keepaspectratio=true,width=7.5cm]{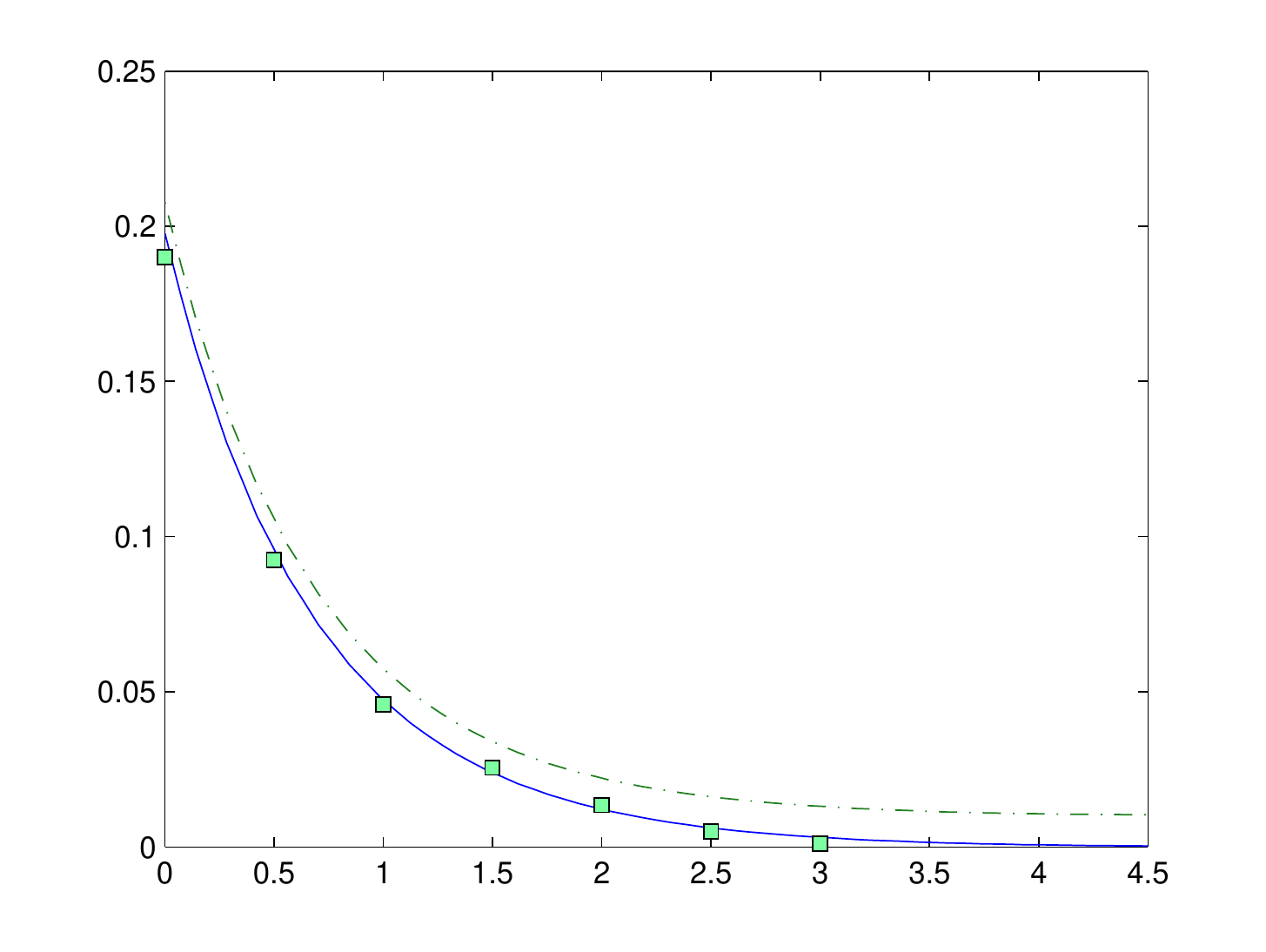}}
\subfigure[Solution with an upper bound from \cite{y1999}]{\includegraphics[keepaspectratio=true,width=7.5cm]{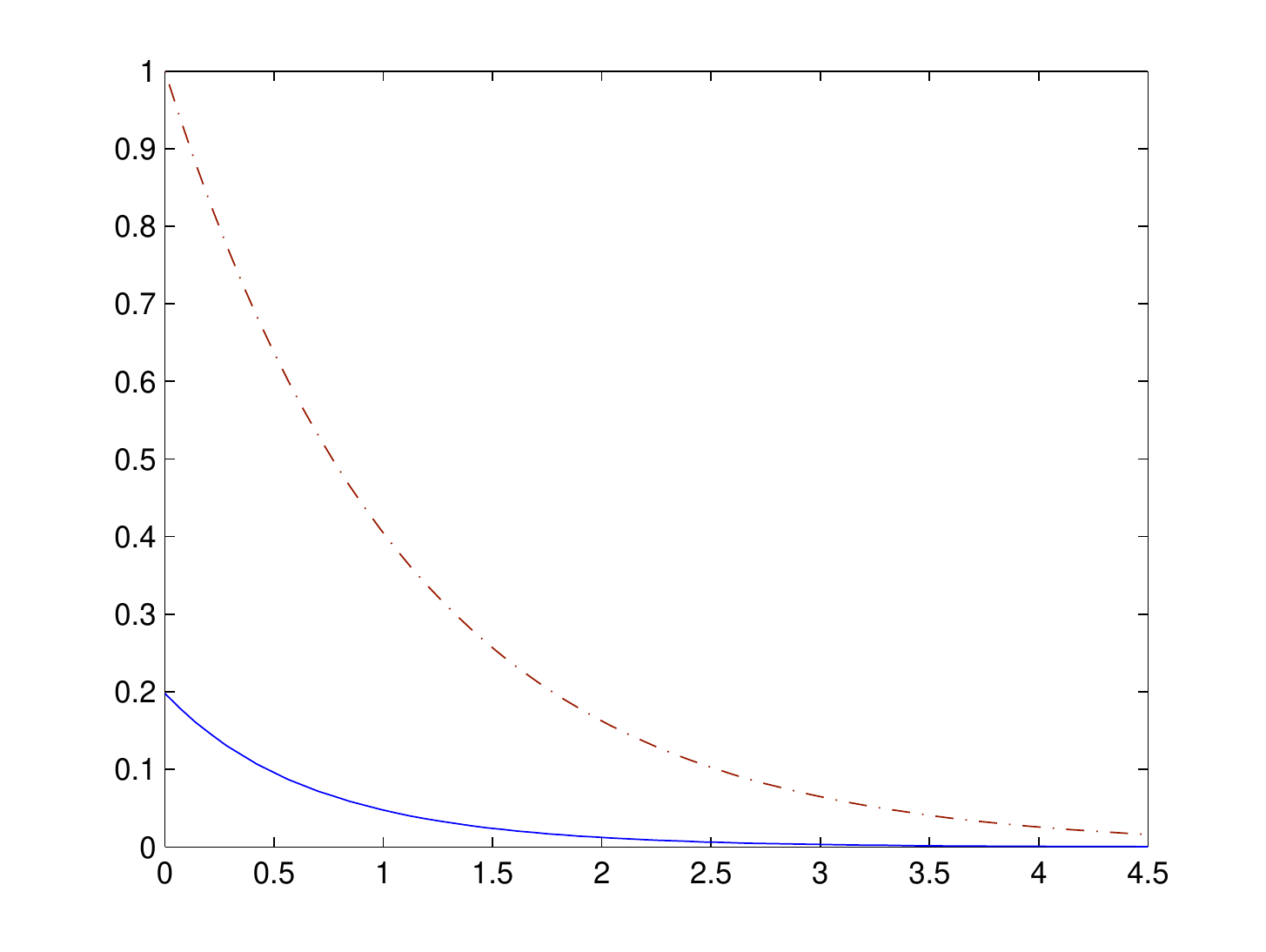}}
\caption{Solution of the ruin probability problem of Section \ref{ssec:CS-ruin} (with interest effect included).
The blue (continuous) line represents the solution of $1-\phi$,
the green (dashed) line in (a) provides an $\ve$-upper bound for the ruin probability $\psi$,
the cyan boxes correspond to results of Monte-Carlo simulations of $\psi$.
In (b), the brown (dashed) line shows the upper bound from \cite{y1999}.}
\label{fig:cs1-interest}
\end{figure}

\subsection{Case study: heavy tails}
\label{ssec:CS-tail}

The following case study is developed to display that the proposed technique works well also in the case of extremely heavy tails.
We again consider a Cramer-Lundberg model for which the distribution of the increment $\eta_0 = G_0-C_0$ is given by the following density function:
\begin{equation}\label{eq:poisson_pdf}
  f_\eta(t) = \frac{\sqrt 2}{\pi\left(1+(t-1)^4\right)}.
\end{equation}
One may notice that $\Ex |\eta|^{3-\delta} < \infty$ for any $\delta\in (0,3]$ but $\Ex |\eta|^3 = \infty$, hence this distribution is heavy-tailed.
Nevertheless, the NPS holds true since $\Ex \eta = 1$,
so we can apply results given in Section \ref{ssec:CL}.

First we find the value of $c$ as per Theorem \ref{thm:korsh}. For this purpose we pick $\gamma = 2$, so $s_1 = 1.07$ and $s_2 = 4$ since $\Ex \eta^2 = 2$. Thus $s_3 = \max\{s_1,s_2\} = 4$, and $c_1 = 2, c_2 = 0.83$ and as a result, $c=5$. Since we have obtained $\gamma = 2$ and $c=5$, in order to obtain an accuracy of $\ve = 0.1$ in \eqref{eq:bounds-korsh} we select $y = 50$.
Recall that we only need to find the value of $\phi(z,50)$ for $z\in [0,50]$. This is obtained by computing a solution of \eqref{eq:fred-phi}, which is a Fredholm equation of the second kind. Here $f_\eta$ is given in \eqref{eq:poisson_pdf} and $F_\eta$ has a closed form:
\begin{equation*}
  F_\eta(x) = \frac{\pi-\arctan(1+\sqrt{2}-x\sqrt{2})+\arctan{(1-\sqrt2+x\sqrt2)}+\arctanh\left(\frac{\sqrt2(x-1)}{1+(x-1)^2}\right)}{2\pi}.
\end{equation*}
The \texttt{FIE} toolbox is employed to numerically solve the integral equation, and the results are displayed in Figure \ref{fig:cs2} for the value of argument in $[0,5]$. Monte Carlo simulations, set up similarly to the previous case study, have been run to validate the results. Again the obtained bounds are much less conservative than the original ones in the literature \cite{k2011}.

\begin{figure}[ht]
  \centering
    \subfigure[Solution with an $\ve$-upper bound]{\includegraphics[keepaspectratio=true,width=7.5cm]{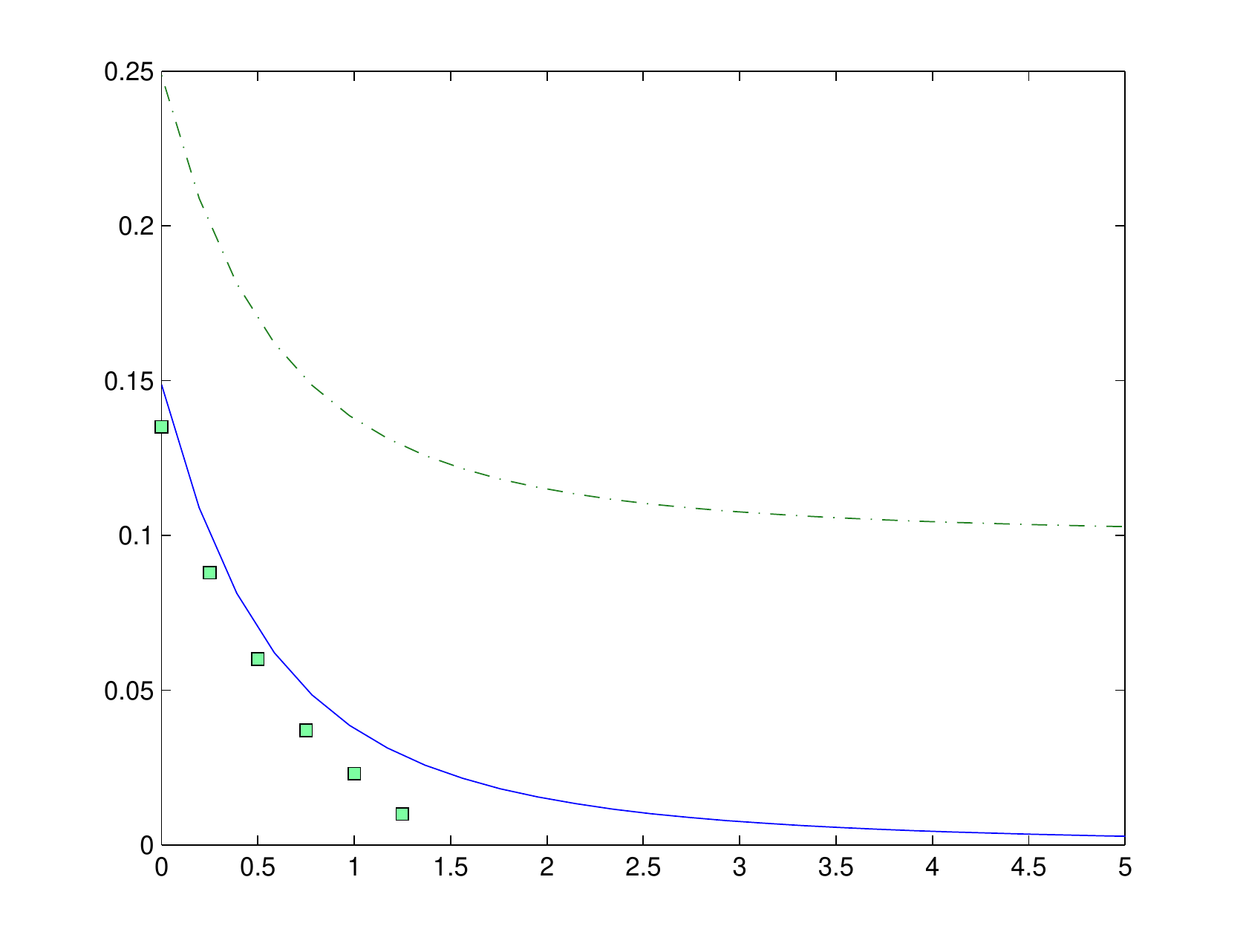}}
    \subfigure[Solution with an upper bound from \cite{k2011}]{\includegraphics[keepaspectratio=true,width=7.5cm]{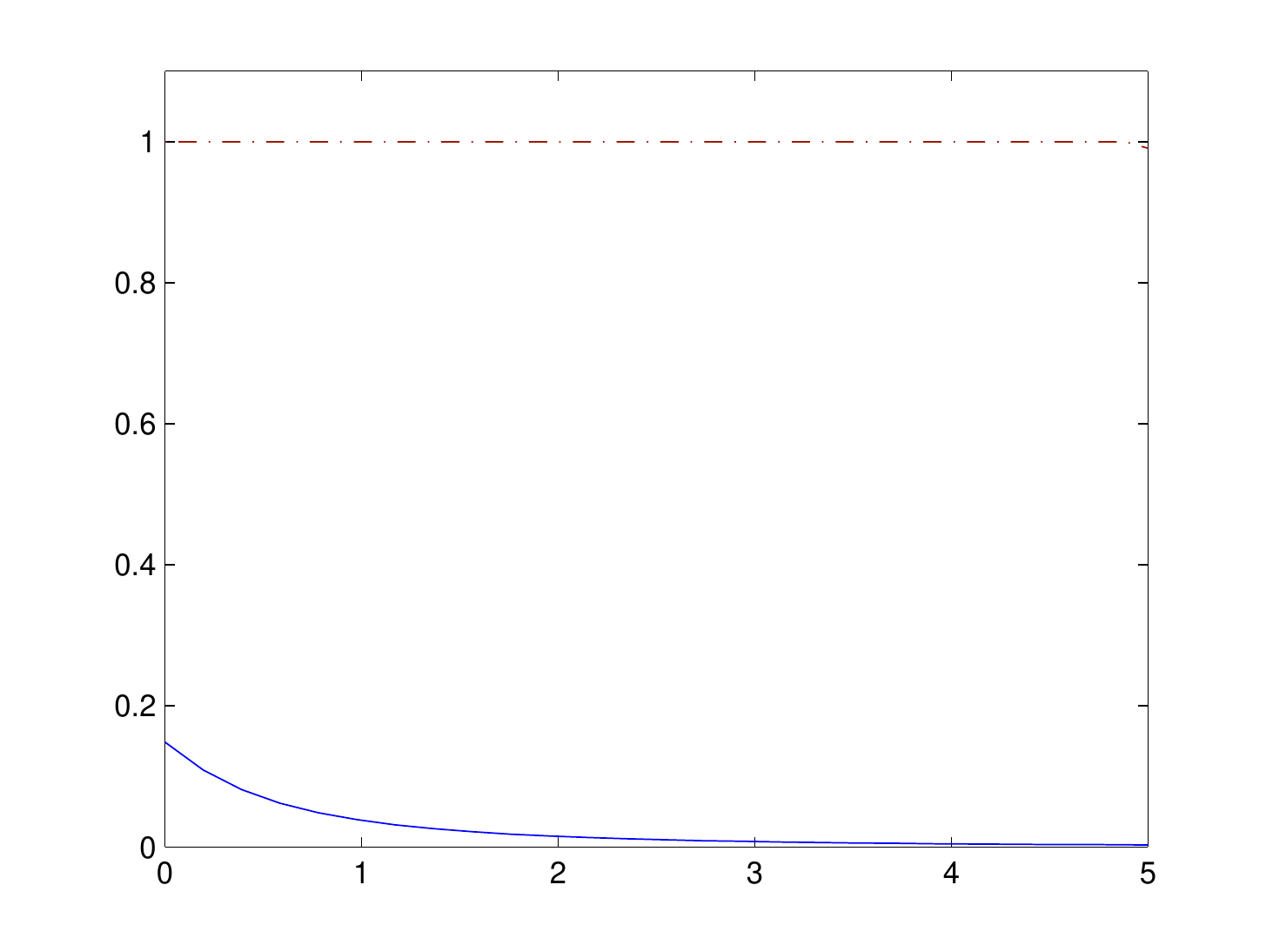}}
  \caption{
      Solution of the tail probability problem. The blue (continuous) line is a plot of $1-\phi$ obtained by integrating \eqref{eq:fred-phi},
      the green (dashed) line in (a) gives an $\ve$-upper bound for the ruin probability $\psi$,
      whereas the cyan boxes correspond to outcomes of Monte-Carlo simulations of $\psi$ used to validate the results.
      In (b), the brown (dashed) line shows the upper bound obtained in \cite{k2011}.
  }
  \label{fig:cs2}
\end{figure}

\section{Conclusions}\label{sec:concl}

This work has discussed the problem of computing the ruin probability for general discrete-time Markov models of risk processes.
The presented approach has highlighted common issues that are shared by a number of models,
and has put forward techniques to overcome such problems.
These methods are based on the relationship between the ruin and two-barrier ruin problems,
and further leverage asymptotic bounds available in the literature.
The generality of the approach makes it possible to deal with diverse complex models in a unified way,
thus allowing considering interest effects or the income from risky investments.

With the goal of elucidating the practical application of the developed results,
this work has provided instances of the implementation of the methods,
together with numerical case studies (with further validation by Monte-Carlo simulations).

The results of this works allow claiming that for model instances (e.g. random walks) with available asymptotic bounds,
the solution of the ruin probability problem can be approximately characterized and computed with any given precision,
is valid globally (for any value of the initial capital) and it improves results (asymptotic bounds) in the literature.

\section*{Acknowledgments}

The authors are grateful to Prof. K. Atkinson for his help on the extremely precise and fast \texttt{FIE} toolbox \cite{as2008}.

\bibliographystyle{alpha}
\bibliography{../../../my_bib}

\end{document}